\documentclass[12pt]{article}
\usepackage{graphicx}
\usepackage{amsmath}
\usepackage{color}
\usepackage{cite}

\newtheorem{theorem}{Theorem}

\newtheorem{definition}[theorem]{Definition}

\newtheorem{lemma}[theorem]{Lemma}

\newenvironment{proof}[1][Proof]{\textbf{#1.} }{\ \rule{0.5em}{0.5em}}
\input{tcilatex}

\begin{document}

\title{Cauchy problem with data on a characteristic cone for the {Einstein-Vlasov}
equations\footnote{Vienna preprint UWThPh-2012-21}}
\author{Yvonne Choquet-Bruhat \and Piotr T. Chru\'sciel}
\maketitle

\begin{center}
\textit{Dedicated to Mario Novello}
\end{center}

\section{Introduction.}

In recent papers~\cite{CCM2,CCM4} (henceforth denoted by I and II) we have
proved existence and uniqueness theorems for solutions of the Cauchy problem
for the Einstein equations in \emph{vacuum} with data on a characteristic
cone $C_{O}$, with vertex at $O$, generalising a construction of~\cite{RendallCIVP,DamourSchmidt}. We have used the tensorial splitting of the Ricci tensor of a Lorentzian
metric $g$ on a manifold $V$ as the sum of a quasidiagonal hyperbolic system
acting on $g$ and a linear first order operator acting on a vector $H,$
called the wave-gauge vector. The vector $H$ vanishes if $g$ is in wave
gauge; that is, if the identity map is a wave map from $(V,g)$ onto $(V,\hat{%
g})$, with $\hat{g}$ some given metric, which we have chosen to be
Minkowski. The data needed for the reduced PDEs is the trace $\bar{g}$ of $g$
on $C_{O},$ but the geometric initial data is the degenerate quadratic form $%
\tilde{g}$ induced by $g$ on $C_{O},$ the missing part of $\bar{g}$ is
determined by a hierarchical system of ordinary differential equations%
\footnote{%
For previous writing of these equations in the case of two intersecting
surfaces in four-dimensional spacetime see Rendall~\cite{RendallCIVP} and
Damour-Schmidt~\cite{DamourSchmidt}.}, called the wave-map-gauge
constraints, along the rays of $C_{O}$, deduced from the contraction of the
Einstein tensor with a tangent to the rays (it is also possible to prescribe
$\tilde{g}$ up to a conformal factor and impose a coordinate condition).

{{The above generalises in a straightforward manner to Einstein equations
coupled to matter fields whose energy-momentum tensor satisfies specific
structure conditions. This is the case for Maxwell fields and scalar fields.
However, it is not immediately clear that kinematic matter sources described
by a Vlasov distribution field $f$ fit into this scheme. In fact, there is a
basic issue arising, related to the fact that the support of the
distribution function $f$ has to be contained within the subset of the
tangent bundle where the momentum of the particles is timelike future
pointing; outside of this region the Vlasov particles are \emph{tachyons. }%
But a \emph{no-tachyons condition} requires \emph{a priori} knowledge of the
whole metric on the initial data hypersurface, while in the wave-map scheme
above only part of the metric is known before the constraints are solved,
the remaining components being determined by ODEs along the generators of
the hypersurface. It has been shown in~\cite{ChPaetz} how to modify the
scheme so that the whole metric can be prescribed on the initial
characteristic hypersurface, which solves that problem, and allows e.g. the
treatment of Vlasov particles with a given mass. However, the formulation in
I and II has the clear advantage, that only geometrically significant
objects are prescribed as free data. This is not the case with~\cite{ChPaetz}%
, as many components of the initial metric have a gauge character. It is
therefore of interest to see whether kinetic matter fits into the original
scheme of [}}7{], as generalised in I and II. In this article we show that
this is indeed the case. More precisely, we prove that} the Cauchy problem
on a characteristic cone for the \emph{Einstein-Vlasov} system can also be
split into a hierarchical system of ordinary differential equations as
constraints and an evolution problem, in the important case of astrophysical
studies where the masses of ``particles'' take a continuous set of positive
values. We show how to construct physically relevant initial values in this
case.

\section{Einstein-Vlasov system.}

\subsection{Distribution function.}

In kinetic theory the matter is composed of a collection of particles whose
size is negligible at the considered scale: rarefied gases in the
laboratory, galaxies or even clusters of galaxies at the cosmological scale.
The number of particles is so great and their motion so chaotic that it is
impossible to observe their individual motions. It is assumed that the state
of the matter in a spacetime $(V,g)$ is represented\footnote{{\footnotesize %
A mathematical justification of the oncoming of chaos in relativistic
dynamics is a largely open problem.}} by a ``one particle distribution
function'':

\begin{definition}
A \textbf{distribution function} is a scalar function $f$ $\geq 0$ on the
tangent bundle $TV$ to the spacetime $V,$%
\begin{equation*}
f:TV\rightarrow R\text{ \ \ by \ \ }(m,p)\mapsto f(m,p_{m}),\text{ \ \ with
\ \ }m\in V\text{, \ \ }p\in T_{m}V.
\end{equation*}
\end{definition}

The physical meaning of the distribution function is that it gives a mean%
\footnote{{\footnotesize In the sense of Gibbs ensemble.}} number of
particles with momentum $p\in T_{m}V$ at a point $m\in V.$

In astrophysics the particles are stars or even galaxies, they do not take a
finite collection of a priori given masses, but rather a continuous family
of positive masses.

One denotes by $\omega _{g}$ and $\omega _{p}$ respectively the volume forms
on $V$ and $T_{m}V$:
\begin{equation*}
\theta =\omega _{g}\wedge \omega _{p}
\end{equation*}
where, in local coordinates
\begin{equation*}
\omega _{g}=(\det g)^{\frac{1}{2}}dx^{0}\wedge dx^{1}\wedge ...\wedge dx^{n},%
\text{ \ \ \ }\omega _{p}:=(\det g)^{\frac{1}{2}}dp^{0}\wedge
dp^{1}...\wedge dp^{n}.
\end{equation*}

\begin{definition}
The moment of order $k$ of the distribution function $f$ is the symmetric $k$%
-tensor
\begin{equation}
T^{\alpha _{1}...\alpha _{p}}(m):=\int_{T_{m}V}f(m,p)p^{\alpha
_{1}}...p^{\alpha _{p}}\omega _{p}.  \notag
\end{equation}
\end{definition}

The moment of order zero, integral on the fiber $T_{m},$ $m\in V$ of the
distribution function:
\begin{equation*}
M(m):=\int_{T_{m}V}f\text{ }\omega _{p}
\end{equation*}
gives the density of presence of particles at a point $m\in $ $V$. The
second moment of $f$ defines the stress energy tensor it generates. It is a
symmetric 2-tensor on $V$ obtained at any point $m\in V$ by integrating on $%
T_{m}$ the product of $f(m,p)$ by the tensor product $p\otimes p.$

\subsection{Vlasov equation.}

In a Lorentzian spacetime $(V,g)$, in the absence of non gravitational
forces and collisions, each particle follows a geodesic of the spacetime
metric $g,$ i.e. an orbit of the vector field $X=(p,P)$ on $TV,$
\begin{equation}
p^{\alpha }:=\frac{dx^{\alpha }}{d\lambda },\text{ \ \ }\frac{dp^{\alpha }}{%
d\lambda }=P^{\alpha },\text{ \ \ with \ \ }P^{\alpha }:=-\Gamma _{\lambda
\mu }^{\alpha }p^{\lambda }p^{\mu }  \notag
\end{equation}
with $\lambda $ a canonical parameter and $\Gamma _{\lambda \mu }^{\alpha }$
the Christoffel symbols of the metric $g.$ Recall that the scalar $g(p,p)$
is constant along an orbit of $X$ in $TV,$ recall also that (Liouville
theorem) the volume form $\theta $ is invariant under the geodesic flow;
that is, $\mathcal{L}_{X}$ denoting the Lie derivative with respect to $X,$
it holds that
\begin{equation*}
\mathcal{L}_{X}\theta =0.
\end{equation*}

In a collisionless model the physical law of conservation of particles
imposes to the form $f\theta $ to be invariant under the vector field $X,$
hence $f\theta $ has a zero Lie derivative with respect to $X.$ Since we
already know that $\mathcal{L}_{X}\theta =0$ the invariance reduces to the
\textbf{Vlasov} equation for $f$
\begin{equation}
\mathcal{L}_{X}f\equiv p^{\alpha }\frac{\partial f}{\partial x^{\alpha }}%
+P^{\alpha }\frac{\partial f}{\partial p^{\alpha }}=0,
\notag
\end{equation}
which says that in the phase space $TV$ the derivative of the distribution
function in the direction of $X$ vanishes

A fundamental theorem  is

\begin{theorem}
If the distribution function $f$ satisfies the Vlasov equation, its moments
have vanishing divergence. In particular
\begin{equation*}
\nabla _{\alpha }T^{\alpha \beta }=0.
\end{equation*}
\end{theorem}

\subsection{\textbf{The} Einstein-Vlasov system.}

Given a smooth manifold $V,$ the unknowns are the metric $g$ on $V$ and the
distribution function $f$ on $TV.$ They must satisfy the Einstein-Vlasov
system
\begin{equation}
S_{\alpha \beta }=T_{\alpha \beta },\text{ \ \ \ on \ }V,\text{ \ \ \ }%
\mathcal{L}_{X}f=0\text{ \ on \ }TV,  \notag
\end{equation}
where $S_{\alpha \beta }$ is the Einstein tensor of $g,$ hence satisfies the
contracted Bianchi identities, $\nabla _{\alpha }S^{\alpha \beta }\equiv 0$
and
\begin{equation*}
T_{\alpha \beta }:=\int_{R^{n+1}}p_{\alpha }p_{\beta }f(.,p)\omega _{p}
\end{equation*}
is the second moment of the distribution function $f,$ hence such that when
the Vlasov equation holds
\begin{equation}
\nabla _{\alpha }T^{\alpha \beta }=0.  \notag
\end{equation}
The Einstein-Vlasov system is therefore coherent.

\section{Cauchy problem on a characteristic cone}

The proof of local existence for solutions of the Cauchy problem with
initial data for $g $ on a spacelike manifold $S$ and for $f$ on $TS$ can be
found in~\cite{YCB:GRbook} and references therein. In what follows we
consider the case where the initial manifold is a characteristic cone of $g.$

The future characteristic cone $C_{O}$ of vertex $O$ for a Lorentzian metric
$g$ is the set covered by future directed null geodesics issued from $O$. We
choose as in I and II coordinates $y^{\alpha }$ such that the coordinates of
$O$ are $y^{\alpha }=0$ and the components $g^{\lambda \mu }(0,0)$ take the
diagonal Minkowskian values, $(-1,1,\ldots ,1)$. If $g$ is Lorentzian and $%
C^{2}$ in a neighbourhood $U$ of $O$ there is an eventually smaller
neighbourhood of $O$, still denoted $U$, such that $C_{O}\cap U$ is an $n$
dimensional manifold, differentiable except at $O$, and there exist in $U$
coordinates $y:=(y^{\alpha })\equiv (y^{0}$, $y^{i}$, $i=1,\ldots ,n)$ in
which $C_{O}$ is represented by the equation of a Minkowskian cone with
vertex $O$,
\begin{equation}
C_{O}:=\{r-y^{0}=0\},\quad r:=\{\sum (y^{i})^{2}\}^{\frac{1}{2}}.
  \notag
\end{equation}
The null rays of $C_{O}$ are represented by the generators of the
Minkowskian cone with tangent vector $\ell .$ We overline as in I and II
traces on $C_{O},$ and underline components in the $y$ coordinates. The
components of $\ell $ are $\underline{\ell ^{0}}=1$, $\underline{\ell ^{i}}%
=r^{-1}y^{i}.$ We denote by $Y_{O}$ the causal future of $O,$ $y^{0}\geq r.$
We denote by $TC_{0}$ the tangent bundle with base $C_{O},$ by $%
T_{g}^{+}C_{O}$ the subbundle with fiber future timelike vectors for the
metric $g.$

\subsection{Cauchy problem for the Vlasov equation with data on $C_{O}$,
given g.}

Let the $n+1$ dimensional spacetime $(V,g)$ be a given Lorentzian manifold.
The Vlasov equation, with $X=X(g),$ and unknown $f$,
\begin{equation*}
\mathcal{L}_{X}f=0
 ,
\end{equation*}
is a first order linear partial differential equation for the distribution
function $f$ on the tangent bundle $TV,$ differential equation on the
geodesic flow which reads in the coordinates $y,$ $\underline{p}$ of $TV$%
\begin{equation}
\frac{df(y(\lambda ),\underline{p}(\lambda ))}{d\lambda }=0
 ,
\notag
\end{equation}
where $y$ and $\underline{p}$ are solutions of the differential system
\begin{equation}
\frac{dy^{\alpha }}{d\lambda }=\underline{p}^{\alpha },\text{ \ \ }\frac{d%
\underline{p}^{\alpha }}{d\lambda }=\underline{P}^{\alpha }\equiv -%
\underline{\Gamma _{\lambda \mu }^{\alpha }}\underline{p}^{\lambda }%
\underline{p}^{\mu }.  \notag
\end{equation}
This quasilinear first order differential system has, if its coefficients
are Lipshitzian i.e. if the considered metric $g$ is $C^{1,1}$, one and only
one solution for $\lambda -\lambda _{0}$ small enough, $(y^{i}(\lambda ),$ $%
y^{0}(\lambda ),\underline{p}^{\alpha }(\lambda ))$, taking for $\lambda
=\lambda _{0}$ given values
\begin{equation}
y^{i}(\lambda _{0})=\bar{y}^{i},\text{ \ }y^{0}(\lambda _{0})=r,
 \text{\ }\underline{p}^{\alpha }(\lambda_{0})=\underline{\pi }^{\alpha },  \notag
\end{equation}
$\underline{\pi }^{\alpha },$ a vector in $R^{n+1}$ timelike for $g.$ The
image in $TV$of the obtained curves
\begin{equation}
\lambda \mapsto (y^{i}(\lambda ),\text{ }y^{0}(\lambda ),\underline{p}%
^{\alpha }(\lambda )),  \notag
\end{equation}
is the set of future timelike geodesics (we mean the curves and their
tangents) issued from points of $T^{+}C_{O}$ in a neighbourhood of $O.$ Well
known properties of timelike geodesics (in particular the conservation of $%
g(p,p)$ under the geodesic flow) show that conversely a past directed
timelike geodesic issued from a point $(m,p)\in $ $T^{+}Y_{O}$, with $m$ in
a small enough neighbourhood of $O,$ will meet $T^{+}C_{O}.$

The initial data for $f$ is a function $\bar{f}$ on $T^{+}C_{O}$. The
solution of the Cauchy problem for the Vlasov equation with this initial
data is given in a neighborhood of $T^{+}O$ in $T^{+}Y_{O}$ by
\begin{equation}
\underline{f}(y(\lambda ),\underline{p}(\lambda ))=\underline{\bar{f}}%
(y^{\alpha }(\lambda _{0}),\underline{p}^{\alpha }(\lambda _{0})).  \notag
\end{equation}

\subsection{Cauchy problem for $g$ with data on $C_{O}$ for the Einstein
equations in wave gauge, given the stress energy tensor of $f$.}

The Einstein tensor in wave gauge forms a quasilinear system of wave
operators deduced from the identities
\begin{equation}
R_{\alpha \beta }^{(h)}(g)\equiv -\frac{1}{2}g^{\lambda \mu }\partial
_{\lambda \mu }^{2}g_{\alpha \beta }+q_{\alpha \beta }(g)(\partial
g,\partial g).  \notag
\end{equation}
The stress energy tensor of a given distribution function $f_{1}$ on a
spacetime $(V,g_{1})$ is an integral operator on $f_{1}$ and $g_{1}$%
\begin{equation*}
T_{1,\alpha \beta }\equiv \int_{R^{n+1}}p_{1,\alpha }p_{1,\beta
}f_{1}(.,p)\omega _{g_{1}}.\text{ \ }
\end{equation*}
One considers the system of quasilinear equations for a metric $g_{2},$ with
$T_{1}$ known,
\begin{equation*}
\text{Einstein}^{(h)}(g_{2})=T_{1} .
\end{equation*}
If $T_{1}$ is smooth there exists (Cagnac - Dossa theorem~\cite{DossaAHP}) $%
t_{0}>0$ such that these equations have a solution $g_{2}$ with trace on $%
C_{O}\cap \{r\in \lbrack 0,t_{0}]\}$ the trace $\bar{g}$ of a smooth spacetime metric.

\subsection{Solution of the Einstein-Vlasov reduced system}

One shows by iteration using known theorems that the Einstein-Vlasov system
in wave gauge \ with unknowns $g$ and $f$ admits a solution with initial
data $\bar{g}$ and $\bar{f} $ if $\bar{g}$ is the trace of a sufficiently smooth
spacetime tensor field  and $\bar f$ is the trace of a sufficiently smooth function on $TV$, compactly supported on each future timelike cone, with for example $\bar{f}=0$ and $\bar{g}$ trace of the
Minkowski metric in a neighbourhood of the vertex.

\subsection{Solution of the original system}

To show that a solution of the reduced system satisfies the original
equations, if the initial data satisfy appropriate constraints, we proceed
as in I where we have proved the decomposition $\ell ^{\alpha }\bar{S}%
_{\alpha \beta }\equiv \mathcal{C}_{\alpha }+\mathcal{L}_{\alpha },$ where $%
\mathcal{C}_{\alpha }$ depends only on $\bar{g}$ and $\mathcal{L}_{\alpha }$
is a linear homogeneous differential operator on the wave gauge vector. The
Bianchi identities show, as in I, that a solution of the reduced system
satisfies the original equations if the initial data satisfy the
constraints, namely
\begin{equation*}
\mathcal{C}_{\alpha }=\ell ^{\alpha }\bar{T}_{\alpha \beta }.
\end{equation*}
We now show that in our astrophysical setting these constraints form again a
hierarchical system of ordinary differential equations if we use coordinates
$x^{\alpha }$ adapted to the null structure of $C_{O}$, defined by
\begin{equation}
x^{0}=r-y^{0},\qquad x^{1}=r\text{ \ \ and }x^{A}=\mu ^{A}(r^{-1}y^{i}),
\notag  \label{3.1}
\end{equation}
$A=2,...n$, local coordinates on the sphere $S^{n-1}$, or angular polar
pseudo coordinates, such that the curves $x^{0}=0$, $x^{A}=$constant are the
null geodesics issued from $O$ with tangent $\ell :=\frac{\partial }{%
\partial x^{1}}.$ The trace $\bar{g}$ on $C_{O}$ of the, a priori general,
Lorentzian metric $g$ that we are going to construct is such that $\overline{%
g}_{11}=0$ and $\bar{g}_{1A}=0;$ we use the notation
\begin{equation}
\bar{g}\equiv \bar{g}_{00}(dx^{0})^{2}+2\nu _{0}dx^{0}dx^{1}+2\nu
_{A}dx^{0}dx^{A}+\bar{g}_{AB}dx^{A}dx^{B}.  \notag
\end{equation}
In particular for the Minkowski metric $\eta \equiv \bar{\eta}$ it holds
that
\begin{equation*}
\eta _{00}=-1,\text{ }\nu _{0}=\nu _{A}=0,\text{ \ }\eta _{AB}=r^{2}s_{AB},
\end{equation*}
where $s_{AB}dx^{A}dx^{B}$ is the metric of the round sphere $S^{n-1}$. The
volume element in the tangent space to $V$ at a point of $C_{O}$ is in the $%
x $-coordinates
\begin{equation*}
\bar{\omega}_{p}\equiv \nu _{0}\det \tilde{g}.
\end{equation*}

\section{Geometric initial data and gauge functions.}

The geometric data is a degenerate quadratic form $\tilde{g}$ induced on $%
C_{O}$ by our unknown $g;$ it reads in coordinates $x^{1}$, $x^{A}$
\begin{equation}
\tilde{g}\equiv \tilde{g}_{AB}dx^{A}dx^{B},  \notag
\end{equation}
i.e.\ $\tilde{g}_{11}\equiv \tilde{g}_{1A}\equiv 0$ while $\tilde{g}%
_{AB}dx^{A}dx^{B}\equiv $ $\bar{g}_{AB}dx^{A}dx^{B}$ is an $x^{1}$-dependent
Riemannian metric on $S^{n-1}.$ While $\tilde{g}$ is intrinsically defined,
it is not so for $\bar{g}_{00}$, $\nu _{0}$, $\nu _{A}$, they are
gauge-dependent quantities, the unknowns of our constraints.

As in the spacelike case, the second fundamental form $\chi $ of $C_{O}$
appears in the constraints, but here it depends only on the induced metric
because the normal to $C_{O},$ $\ell \equiv =\frac{\partial }{\partial x^{1}}%
,$ is also tangent. We have defined $\chi $ in I as the tensor $\chi $ on $%
C_{O}$ given by the Lie derivative $\mathcal{L}_{\ell }\tilde{g}$ with
respect to $\ell $ of the degenerate quadratic form $\tilde{g}$, namely in
the coordinates $x^{1},x^{A}:$
\begin{equation*}
\chi _{AB}\equiv \frac{1}{2}\partial _{1}\bar{g}_{AB},\text{ \ \ }\chi
_{A1}\equiv \chi _{11}\equiv 0.
\end{equation*}

We denote by $\tau $ the mean extrinsic curvature $\tau :=g^{AB}\chi _{AB}.$

\section{The first constraint operator}

We have found in I the identity
\begin{equation*}
\bar{\ell}^{\beta }\bar{S}_{1\beta }\equiv \overline{R}_{11}\equiv -\partial
_{1}\tau +\nu ^{0}\partial _{1}\nu _{0}\tau -\frac{1}{2}\tau (\overline{%
\Gamma }_{1}+\tau )-\chi _{A}{}^{B}\chi _{B}{}^{A}.
\end{equation*}
By definition of the wave-gauge vector $H$ we have
\begin{equation*}
\overline{\Gamma }_{1}\equiv \overline{W}_{1}+\overline{H}_{1},\text{ \ \ }%
\overline{W}_{1}=-\nu _{0}x^{1}\overline{g}^{AB}s_{AB},
\end{equation*}
and hence $\overline{R}_{11}$ decomposes as
\begin{equation*}
\overline{R}_{11}\equiv \mathcal{C}_{1}+\mathcal{L}_{1}\;,\text{\ where\
}\mathcal{L}_{1}:=-\frac{1}{2}\overline{H}_{1}\tau \;\text{vanishes in wave
gauge},
\end{equation*}
while $\mathcal{C}_{1}$\ involves only the values of the coefficients $%
\overline{g}_{AB}$ and $\nu _{0}$ on the light-cone; the first constraint
reads
\begin{equation*}
\mathcal{C}_{1}:=-\partial _{1}\tau +\{\nu ^{0}\partial _{1}\nu _{0}-\frac{1%
}{2}(\overline{W}_{1}+\tau )\}\tau -\chi _{A}{}^{B}\chi _{B}{}^{A}=\bar{T}%
_{11}\equiv (\nu _{0})^{2}\bar{T}^{00},
\end{equation*}
\begin{equation*}
\bar{T}^{00}\equiv \det \tilde{g}\int_{R^{n+1}}(p^{0})^{2}\bar{f}%
(.,p)d^{n+1}p,\text{ \ \ }d^{n+1}p=dp^{0}dp^{1}...dp^{n}.
\end{equation*}
The first constraint contains $\nu _{0}$ as only unknown if $\tilde{g}$ and $%
\bar{f}$ are given. The resulting ODE is singular for $\tau =0$ since $\tau $
multiplies $\partial _{1}\nu _{0}.$ However the alternative used by Rendall~\cite{RendallCIVP} and Damour and Schmidt~\cite{DamourSchmidt} of prescribing only the conformal class of $%
\tilde{g},$ splitting $\chi $ into its conformally invariant traceless part $%
\sigma $ and its unknown trace $\tau $ and imposing to $\nu _{0}$ to annul
the parenthesis in  $\mathcal{C}_{1}$ does not work here because $\bar{%
T}_{11}$ depends upon $\nu _{0}.$

We look for a solution of the first constraint tending to 1 as $r$ tends to
zero, hence set $\nu _{0}=1+u$ and write the first constraint as the first
order differential equation (recall that $\nu ^{0}=(\nu _{0})^{-1})$
\begin{equation*}
r\partial _{1}u=r(a+b)+\beta +(a+2b+3\beta )u+(b+3\beta )u^{2}+\beta u^{3},
\end{equation*}
where
\begin{equation}
a:=\tau ^{-1}\partial _{1}\tau +\frac{1}{2}\tau +\tau ^{-1}|\chi |^{2},\text{
\ }b:=-\frac{1}{2}\bar{g}^{AB}rs_{AB},\text{ \ }\beta \equiv r\tau ^{-1}\bar{%
T}^{00}.  \notag
\end{equation}
To avoid analytical difficulties near the vertex, we assume that there
exists a neighborhood $0<r\leq r_{0}$ of the vertex $O$ within $C_{O}$ on
which the initial data $\bar{f}$ vanishes and $\tilde{g}$ coincides with $%
\tilde{\eta}$. In this neighbourhood the coefficients reduce to their
Minkowskian values
\begin{equation*}
a_{0}=\frac{n-1}{2r},\text{ \ }b_{0}:=-\frac{n-1}{2r},\text{ \ }\beta _{0}=0,
\end{equation*}
the solution $u$ tending to zero as $r$ tends to zero is then $u=0$ for $%
r\leq r_{0}.$ For $r>r_{0}$ it is the solution $u$ vanishing for $r=r_{0}$
of a smooth equation with smooth coefficients, equivalently of the integral
equation
\begin{equation*}
u(r,x^{A})=\int_{r_{0}}^{r}\{(a+b)+(x^{1})^{-1}[\beta +(a+2b+3\beta
)u+(b+3\beta )u^{2}+\beta u^{3}]\}dx^{1}.
\end{equation*}
This solution can be computed by the Picard iteration method as long as the
interval $[r_{0},r]$ is small enough; its size, and the norm of $u$ depend
on the size of the various coefficients. We will always have $\nu _{0}>0$ in
a small enough neighbourhood, as required for the Lorentzian character of $%
\bar{g}$.

\section{The $C_{A}$ constraint}

We have written in I the $\mathcal{C}_{A}$ constraint operator. The $C_{A}$
constraint with kinetic source reads
\begin{equation}
\mathcal{C}_{A}-\bar{T}_{1A}\equiv -\frac{1}{2}(\partial _{1}\xi _{A}+\tau
\xi _{A})+\tilde{\nabla}_{B}\chi _{A}^{B}-\frac{1}{2}\partial _{A}\tau
+\partial _{A}(\frac{1}{2}\bar{W}_{1}+\nu _{0}\partial _{1}\nu ^{0})-\bar{T}%
_{1A}=0  ,
\notag
\end{equation}
where $\xi _{A}$ is defined as
\begin{equation}
\xi _{A}:=-2\nu ^{0}\partial _{1}\nu _{A}+4\nu ^{0}\nu _{C}\chi
_{A}^{C}+\left( \bar{W}^{0}-\frac{2}{r}\nu ^{0}\right) \nu _{A}+\bar{g}_{AB}%
\bar{g}^{CD}(S_{CD}^{B}-\tilde{\Gamma}_{CD}^{B})  \notag
,
\end{equation}
while
\begin{equation}
\bar{T}_{1A}(x)\equiv \int_{R^{n+1}}p_{1}p_{A}F(x,p)d^{n+1}p=\nu
_{0}\{\int_{R^{n+1}}p^{0}(\nu _{A}p^{0}+g_{AB}p^{B})d^{n+1}p  \notag
.
\end{equation}
When $\nu _{0}$ has been determined the only unknown in $\mathcal{C}_{A}-%
\bar{T}_{1A}$ is $\nu _{A}.$ It satisfies a linear equation, its solution
follows the same lines as in I.

\section{The $C_{O}$ constraint}

When $\nu _{0}$ and $\nu _{A}$ have been determined, the $\mathcal{C}_{0}$
constraint operator, $\ell ^{\beta }\bar{S}_{0\beta }\equiv \bar{S}_{01},$
contains as only unknown $ \bar{g}_{00}\equiv \underline{\bar{g}_{00}},$ we
have found in $I$ and $II$ after long computations that the $\mathcal{C}_{0}$
constraint can be written, using the other constraints
\begin{equation}
\partial _{1}\zeta +(\kappa +\tau )\zeta +\frac{1}{2}\{\partial _{1}\bar{W}%
^{1}+(\kappa +\tau )\bar{W}^{1}+\tilde{R}-\frac{1}{2}g^{AB}\xi _{A}\xi _{B}+%
\overline{g}^{AB}\tilde{\nabla}_{A}\xi _{B}\}+h_{f}=0,  \notag
\end{equation}
\begin{equation}
h_{f}:=\frac{1}{2}\bar{g}^{11}\bar{T}_{11}+\bar{g}^{1A}\bar{T}_{1A}+\bar{g}%
^{10}\bar{T}_{10},  \notag
\end{equation}
\begin{equation*}
\kappa \equiv \nu ^{0}\partial _{1}\nu _{0}-\frac{1}{2}(\bar{W}_{1}+\tau
),\quad \bar{W}_{1}\equiv \nu _{0}\bar{W}^{0},\quad \bar{W}^{0}\equiv \bar{W}%
^{1}\equiv -r\bar{g}^{AB}s_{AB}.
\end{equation*}
The relation between $\zeta \,$and $\bar{g}_{00}$ is
\begin{equation}
\zeta :=(\partial _{1}+\kappa +\frac{1}{2}\tau )\bar{g}^{11}+\frac{1}{2}\bar{%
W}^{1},  \notag
\end{equation}
where $\bar{g}^{11}$ and $\bar{g}_{00}$ are linked by the linear relation
\begin{equation}
\underline{\bar{g}_{00}}\equiv \bar{g}_{00}\equiv -\bar{g}^{11}(\nu
_{0})^{2}+\bar{g}^{AB}\nu _{B}\nu _{A}.  \notag
\end{equation}
The component $\bar{T}_{01}$ of the stress energy tensor of the distribution
function reads
\begin{equation}
\bar{T}_{01}(\bar{x})\equiv \int_{T_{m}}p_{0}p_{1}\bar{f}(\bar{x}%
,p)d^{n+1}p\equiv \nu _{0}\int_{T_{m}}(\bar{g}_{00}p^{0}+\nu _{0}p^{1}+\nu
_{A}p^{A})p^{0}\bar{f}(\bar{x},p)\text{ }d^{n+1}p.  \notag
\end{equation}
The only unknown remaining in the $\mathcal{C}_{0}$ constraint is $\bar{g}%
_{00}.$ It appears linearly. The solution equal to $-1$ in $[0,r_{0}]$ is
smooth and negative in $[0,r_{0}+\varepsilon ] .$

\section{The tachyon problem.}

For a physical reason, namely the positivity of the masses of the considered
``particles'', the initial data $\bar{f}$ of the distribution function must
have its support in the subset
\begin{equation}
\bar{g}_{00}(p^{0})^{2}+2\nu _{0}p^{0}p^{1}+2\nu _{A}p^{0}p^{A}+\bar{g}%
_{AB}p^{A}p^{B}<0.  \tag{1}
\end{equation}
Since only $\bar{g}_{AB}$ has been a priori given the above inequality gives
an a posteriori restriction on the support of $\bar{f}.$ It is natural to
assume that this support in momentum space at the tip, where all components
of the metric are known, is contained within the future Minkowski timelike-cone.
Then, by continuity of the data $\bar{f}$\ and of the solutions of the
constraints, one obtains some neighborhood of the tip in $C_{O}$\ where this
support will remain in the future cone of the constructed metric. This
neighborhood will provide the relevant initial data. A precise statement to
this effect can be obtained as follows\textbf{:}

Recall that $\underline{p}^{0}\equiv p^{0}+p^{1}$ is the time component of $%
p $ in the $y$ coordinates, the condition that we will give is that $p$ is
strictly timelike for a known metric on the support of $\bar{f}$. The
following lemma is a result of the continuity properties found for $\bar{g}$:

\begin{lemma}
Assume that there exists $r_{0}\geq 0$ such that the support of the data $%
\bar{f}$ in $TC_{O}$ is contained in the subset
\begin{equation}
(\underline{p}^{0})^{2}>k\{\bar{g}_{AB}p^{A}p^{B}+(p^{1})^{2}\},\text{ \ \ \
}k>1\text{ \ \ and \ }p^{0}>0  \tag{2}
\end{equation}
with $\bar{g}_{Ab}=\eta _{AB}$ for $r\leq r_{0}$ (possibly $r_{0}=0).$ Then
there is a neighbourhood $r_{0}\leq r<R$ in $C_{O\text{ }}$such that the
support of $\bar{f}$ in $TC_{O}\cap \{r_{0}\leq r<R\}$ includes no tachyons
for the metric $\bar{g}.$
\end{lemma}

\begin{proof}
The no-tachyon (positive masses) condition can be written with $\nu
_{0}=1+u, $ $g^{00}=-1+\alpha $ where $u$, $\alpha $ and $\nu _{A}$ tend to
zero as $r $ tends to $r_{0}$ (possibly zero)
\begin{equation}
(\underline{p}^{0})^{2}-(p^{1})^{2}-\alpha (p^{0})^{2}+2p^{0}(up^{1}+\nu
_{A}p^{A})-\bar{g}_{AB}p^{A}p^{B}>0.  \tag{3}
\end{equation}
Remark that $p^{0}+p^{1}\equiv \underline{p}^{0}$ implies
\begin{equation}
(p^{0})^{2}\leq 2\{(\underline{p}^{0})^{2}+(p^{1})^{2}\},\text{ \ }%
2p^{0}p^{1}\leq (p^{0})^{2}+(p^{1})^{2}\leq 2(\underline{p}%
^{0})^{2}+3(p^{1})^{2},  \tag{4}
\end{equation}
while setting $|\tilde{\nu}|:=(\bar{g}^{AB}\nu _{A}\nu _{^{B}})^{\frac{1}{2}%
} $ we have
\begin{equation}
|\nu _{A}p^{A}|\leq (\bar{g}_{AB}p^{A}p^{B})^{\frac{1}{Z}}|\tilde{\nu}|\text{
\ hence }2p^{0}\nu _{A}p^{A}\leq \text{\ }|\tilde{\nu}|\{(p^{0})^{2}+\bar{g}%
_{AB}p^{A}p^{B}\} .  \tag{5}
\end{equation}
The inequalities (4) and (5) \ show that the no-tachyon condition (3) is
implied by the following inequality
\begin{equation}
(\underline{p}^{0})^{2}-(2|\alpha |+4|u|+|\tilde{\nu}|)(\underline{p}%
^{0})^{2}>(p^{1})^{2}+\bar{g}_{AB}p^{A}p^{B}+3|u|(p^{1})^{2}+|\tilde{\nu}%
|\{(p^{0})^{2}+\bar{g}_{AB}p^{A}p^{B}.  \tag{6}
\end{equation}

The assumption (2) of the lemma (which depends only on the given data $%
\tilde{g})$ implies that there exist two numbers $k_{1}<1$ and $k_{2}>1$
with $k_{1}^{-1}k_{2}=k$ such that
\begin{equation}
k_{1}(\underline{p}^{0})^{2}>k_{2}\{\bar{g}_{AB}p^{A}p^{B}+(p^{1})^{2}\},%
\text{ \ \ \ }k>1\text{ \ \ and \ }p^{0}>0.  \tag{7}
\end{equation}
The properties of $\alpha ,u$ and $\tilde{\nu}$, which tend to zero as $r$
tends to $r_{0}$, show that there exists $R>r_{0}$ such that for $r_{0}\leq
r\leq R$ it holds that
\begin{equation*}
(\underline{p}^{0})^{2}-(2|\alpha |+4|u|+|\tilde{\nu}|)(\underline{p}%
^{0})^{2}>k_{1}(\underline{p}^{0})^{2},
\end{equation*}
\begin{equation*}
(p^{1})^{2}+\bar{g}_{AB}p^{A}p^{B}+3|u|(p^{1})^{2}+|\tilde{\nu}|\bar{g}%
_{AB}p^{A}p^{B}<k_{2}\{\bar{g}_{AB}p^{A}p^{B}+(p^{1})^{2}\}.
\end{equation*}
These inequalities complete the proof of the lemma.
\end{proof}

\section{Conclusion}

From a smooth solution of the constraints in the $x$ variables, Minkowskian
for $0\leq r<r_{0},$ one deduces that $\bar{g}$ is the trace on $C_{O}$ of a
smooth spacetime function\textbf{, }which completes the proof of the
existence of a solution of an Einstein-Vlasov spacetime in a neighbourhood
of $O$ taking the given data $\tilde{g},\bar{f}.$ Geometric uniqueness is
also easy to prove. The existence up to the vertex is more involved: the
technique of admissible series used in II could be used, but requires
further work.

\end{document}